\documentclass[runningheads]{llncs}
\usepackage{graphicx}
\usepackage{todonotes}
\usepackage{amsmath}
\usepackage{amssymb}
\usepackage{amsfonts}
\usepackage{fontawesome}
\usepackage{booktabs}
\usepackage{paralist}
\usepackage[smaller]{acronym}
\usepackage[ruled,vlined,english]{algorithm2e}
\usepackage[nocompress]{cite}
\usepackage{threeparttable}
\usepackage{listings}
\newcommand{\OPT}{\operatorname{\text{\textsc{opt}}}}

\newcommand{\PROFIT}{\operatorname{\text{\textsc{profit}}}}
\newcommand{\AREA}{\operatorname{\text{\textsc{area}}}}
\newcommand{\WEIGHT}{\operatorname{\text{\textsc{weight}}}}
\newcommand{\SOL}{\operatorname{\text{\textsc{sol}}}}

\let\doendproof\endproof
\renewcommand\endproof{~\hfill\qed\doendproof}
\begin{document}

\title{Robust Online Algorithms for Dynamic Choosing Problems}
%
%
\author{Sebastian Berndt\inst{1}, Kilian Grage\inst{2}\thanks{Supported by DFG-Project JA 612 /19-1 and GIF-Project "Polynomial Migration for Online Scheduling"}, Klaus Jansen\inst{2},\\ Lukas Johannsen\inst{2}, Maria Kosche\inst{3}}
\authorrunning{S. Berndt et al.}
%
\institute{
University of L\"ubeck, 23562  L\"ubeck, Germany
\\\email{s.berndt@uni-luebeck.de}
\and
Kiel University, 24118 Kiel, Germany\\
\email{\{kig,kj\}@informatik.uni-kiel.de}
\and
G\"ottingen University, 37073 G\"ottingen, Germany
\email{maria.kosche@cs.uni-goettingen.de}}
\maketitle              
\begin{abstract}
  Semi-online algorithms that are allowed to perform a bounded amount of repacking 
  achieve guaranteed good worst-case behaviour in a more realistic setting. 
  Most of the previous works focused on minimization problems that aim to
  minimize some costs.  In this work, we study maximization problems that aim to maximize their
  profit.
	
	We mostly focus on a class of problems that we call \emph{choosing problems}, where a maximum profit subset of a set objects has to be maintained. Many known problems, such as \textsc{Knapsack}, \textsc{MaximumIndependentSet} and variations of these, are part of this class. We present a framework for choosing problems that allows us to transfer offline $\alpha$-approximation
  algorithms into $(\alpha-\epsilon)$-competitive semi-online algorithms with amortized migration $O(1/\epsilon)$. Moreover we complement these positive results with lower bounds that show that our results are tight in the sense that no amortized migration of $o(1/\epsilon)$ is possible. 

\end{abstract}

\section{Introduction} 

Optimization problems and how fast we can solve them optimally or
approximatively have been a central topic in theoretical computer science. 
These kind of problems usually have their origins in the real world and solving
them in most cases is not only relevant from a theoretical perspective but also
has many applications. 
These optimization problems do not account for one major
problem that is unique to applications: 
an unknown future. 
Usually,  we are not given all
the information in advance, as unforeseeable things like customers cancelling or
new urgent customer requests can happen at any moment. 
This context gave rise to the study of \emph{online problems} in different
variants to model this uncertainty. 
The classical model starts with an empty instance and in subsequent time
steps, new parts of the instance are added. 
In order to solve a problem, an online algorithm must generate a solution for
every time step without knowing any information about future events. \looseness=-1

In the strictest setting, the algorithm is not allowed to alter the solution
generated in a previous step at all, so every mistake will carry  weight into the
future. 
As this is a very heavy restriction for online algorithms, there are also
variants where the algorithm is allowed to change solutions to some degree. 
We cannot allow for an arbitrary number of changes as this only leads to the
offline setting. 
We therefore consider the \emph{migration model}, where every change in the
instance, e.\,g.~an added node to a graph or a new item, comes
with a \emph{migration potential}. 
Intuitively, this migration potential is linked to some size or weight which
means objects that have a larger impact on the optimization criteria
will yield larger migration potential allowing more change. 
Similarly, small objects will only allow for small changes of the
solution. 
The ratio between the sum of changed objects in the solution and the total
migration potential at a time is the so called \textit{migration factor}. \looseness=-1

We  consider this migration setting in an amortized way by allowing the migration to accumulate
over time. 
This way we allow our algorithm to generally handle newly arriving objects without
changing the solution (apart from extending the solution with regards to the new
objects) and at some later point of time we will repack the solution and use the
sum of all migration potential of items that arrived up to that time. 
On the matter of (amortized) migration there are two criteria to consider. 
For one, as usual in optimization theory, we want to achieve good competitive
ratios, meaning a solution close to an optimal offline solution. 
On the other hand, we would like that the  migration also remains small. 
In general this is quite an intricate question as both these factors counteract
each other. 
The better the solution we generate, the more we need to repack and vice versa. 
Despite this duality, we will present in this paper a very simple framework that
achieves results close to the best offline results for a large range of
problems. 
In fact, we manage to keep solutions on par with the best offline algorithms
except an additive $\epsilon$-term. 
Surprisingly, we only need an amortized migration factor of
$O(1/\epsilon)$ despite maintaining such a high quality solution. 
For many problems this framework even works when considering the problem
variants where objects not only appear but are also removed from the instance.
In addition to these positive results, we also show that a migration of only
$\Omega(1/\epsilon)$ is needed even for relatively simple problems such as the 
\textsc{SubsetSum} problem.

\section{Preliminaries}

We are given some optimization problem $\Pi_{\textrm{off}}$ consisting of a set
of objects and consider the online version $\Pi_{\textrm{on}}$, where these
objects arrive one by one over time (the \emph{static} case).
If arrived items can also be removed from the instance, we call this the
\emph{dynamic} case. 
For an instance $I\in \Pi_{\textrm{on}}$ and some time $t$ we denote the
instance at time $t$ containing the first $t$ objects by $I_t$.
As discussed above, we allow a certain amount of repacking and thus, every
object has an associated \emph{migration potential}, which typically corresponds
to its size or its weight. \looseness=-1

The item arriving or departing at time $t$ has migration potential
$\Delta(I_{t})$.
Generally in the following, we will write $\Delta(I_{t}\to I_{t'})$ to denote the
migration potential that we received starting at time $t$ until $t'$ for $t\le
t'$. If the given instance is clear from the context, we will simplify this
notation and write $\Delta_{t:t'} := \Delta(I_{t}\to I_{t'})$. The total
migration potential up to some time $t$ is thus given by $\Delta_{0:t}$.

We further assume that for every two feasible solutions $S_t$ and $S'_{t'}$ at times $t,t'$,
we also have a necessary migration cost
that we denote by $\phi(S_t\to S'_{t'})$, respectively.
This resembles the costs to migrate from solution $S_{t}$ to solution $S'_{t}$.
We assume that these costs can be computed easily by comparing the respective
solutions. Note that the initial assignment of an item does not cost any migration.
Hence, we often write $\phi(S_{t}\to S_{t+1})$ to denote the migration cost of changing solution
$S_{t}$ to solution $S_{t+1}$ and assume that the newly arrived item is now also
present in~$S_{t+1}$.

We say that an online algorithm has \emph{amortized migration factor} $\gamma$
if, for all time steps $t$, the sum of the migration costs is at most
$\gamma\cdot \sum_{i=1}^{t}\Delta(I_{t}) = \gamma\cdot \Delta_{0:t}$, i.\,e.
  $\sum_{i=1}^{t}\phi(S_{i-1}\to S_{i}) \leq \gamma\cdot \Delta_{0:t}$,
where $S_{i}$ are the solutions produced by the algorithm with $S_0$ being the
empty solution. 
We will sometimes make use of the amortized migration factor inside a time interval $t \to \cdots \to t'$, which we will denote by $\gamma_{t:t'}$.
The notions of \emph{migration} or \emph{migration factor} are interchangeably
used by us and always describe the amortized migration factor.

In this work, we only consider maximization problems.
Hence, every solution $S_{t}$ of an instance $I_{t}$ has some profit
$\PROFIT(S_{t})$ and $\OPT(I_{t})$ denotes the optimal profit of any solution to
$I_{t}$.
An algorithm that achieves competitive ratio $1-\epsilon$ and has migration
$f(1/\epsilon)$ for some function $f$ is called
\emph{robust}~\cite{DBLP:journals/mor/SandersSS09}. 
Due to page restrictions, some of our results can be found in the appendix.

\subsection{\textsc{Knapsack}-type problems}

The \textsc{Knapsack} problem is one of the classical maximization problems.
In its most basic form,
it considers a \emph{capacity} $C\in \mathbb{N}$
and a finite set $I$ of \emph{items},
each of which is assigned a \emph{weight} $w_i \in \mathbb N$ and a \emph{profit} $p_i \in \mathbb N$.
The objective is to find a subset $S \subseteq I$,
interpreted as a packing of the figurative knapsack,
with maximum profit $\PROFIT(S) = \sum_{i \in S} p_i$ while the total weight $\WEIGHT(S) = \sum_{i \in S} w_i$ does not exceed $C$.
The special case in which all weights are equal to their respective profits is the \textsc{SubsetSum} problem.
In this case, because weight and profit coincide,
we will simply call both the \emph{size} of an item.
A natural generalization of \textsc{Knapsack} is to generalize capacity and weight vector to be $d$-dimensional vectors,
i.e. $C,w_i \in \mathbb{N}^d$ for some $d\in \mathbb{N}$.
The problem of finding a maximum profit packing fulfilling all $d$ constraints
is then known as the $d$-dimensional \textsc{Knapsack}~problem.\looseness=-1

In the \textsc{MultipleKnapsack} problem, one is given an instance $I$
consisting of items with assigned weights and profits,
just like in the  \textsc{Knapsack} problem.
However, in contrast, we are given not one but $m$ different knapsacks with respective capacities $C^{(1)},...,C^{(m)}$.
The goal is to find $m$ disjoint subsets $S^{(1)},...,S^{(m)}$ of $I$,
such that the total profit $\sum_{j=1}^m \PROFIT(S^{(j)})$ is maximized w.r.t. the capacity conditions
\begin{equation*}
\WEIGHT(S^{(j)}) \leq C^{(j)}, \quad \forall j = 1,...,m.
\end{equation*}
Note that the \textsc{Knapsack} problem is a special case of the \textsc{MultipleKnapsack} problem with $m = 1$.

Another generalization of the standard \textsc{Knapsack} problem is
\textsc{2DGeoKnapsack}.
This problem takes as input the width $W\in \mathbb{N}$ and height $H\in \mathbb{N}$ of the knapsack
and a set $I$ of axis-aligned rectangles $r \in I$
with widths $w_r\in \mathbb{N}$, heights $h_r\in \mathbb{N}$, and profits $p_r\in \mathbb{N}$.
An optimal solution to this instance consists of a subset $S \subseteq I$ of the rectangles
together with a non-overlapping axis-aligned packing of $S$ inside the rectangular knapsack of size $W \times H$
such that $\PROFIT(S)$ is maximized.\looseness=-1


\subsection{Independent Set}
Another classical optimization problem is the \textsc{MaximumIndependentSet} problem.
While it can be considered for different types of graphs,
the most basic variant is defined on a graph $G = (V, E)$ with a set of nodes $V$
and a set of corresponding edges $E$.
A subset $S \subseteq V$ is called \emph{independent}
if for all $u,v \in S$ it holds that $(u,v) \notin E$, i.e. $u$ and $v$ are not neighbours.
A maximal independent set is then an independent set
that is no strict subset of another independent set.
\textsc{MaximumIndependentSet} is the problem of finding a maximum independent set for a given graph~$G$.
In the literature,
\textsc{MaximumIndependentSet} is,
among others, studied in planar, perfect, or claw-free graphs.
For the online variant,
we usually assume that a node is added (or removed) to the instance in every time step along with its adjacent edges.\looseness=-1

Closely related to the well-studied \textsc{MaximumIndependentSet} problem is
the \textsc{MaximumDisjointSet} problem.
For a given instance $I$ that consists of items with a geometrical shape,
the goal is to find a largest disjoint set
which is a set of non-overlapping items.
As we can convert an \textsc{MDS} instance to an \textsc{MIS} instance,
we sometimes use \textsc{MIS} to also denote this problem.
\textsc{MDS} is often considered limited to certain types of objects.
These can be (unit-sized) disks, rectangles, polygons or other objects,
and any $d$-dimensional generalization of them.
We will also consider pseudo-disks,
which are objects that pairwise intersect at most twice in an instance.

The standard \textsc{MIS} and \textsc{MDS} problems
are both special cases of the generalized \textsc{MaxWeightIndependentSet} or \textsc{MaxWeightDisjointSet}, respectively,
where each node $i$ is assigned a profit value $w_i$.
The goal for these problems is to find a maximum profit independent subset.

\subsection{Our Results}
Our main result is a framework that is strongly inspired by the approach of Berndt \emph{et al.}~\cite{DBLP:conf/waoa/BerndtDGJK19}\textbf{}.
For minimization problems, they proposed a framework using two known algorithms,
one online and one offline algorithm,
in order to solve a given problem.
This approach behaves a bit differently for maximization problems in terms of the competitive ratio.
While a respective $\alpha$-approximation offline algorithm paired with a fitting $\beta$-competitive online algorithm yields an $\alpha + O(1) \beta \epsilon$ competitive-algorithm for minimization problems,
we show that for maximization problems such fitting algorithms will result in a $\alpha \cdot \beta$-competitive algorithm.
The general analysis appears in appendix $\ref{GenFW}$.\looseness=-1

In this work we discuss a class of problems that is characterized by the common task of choosing a subset of objects with maximum profit fulfilling some secondary constraints.
Many important problems like the above presented variants of \textsc{Knapsack}
or \textsc{MaximumIndependentSet} are covered by this class of problems. We show that for these kind of problems, which we will call \emph{choosing problems} in the following,
the framework can be simplified by completely removing the online algorithm.


Using a simplified framework, we achieve $(1-\epsilon)$-competitive algorithms for \textsc{Knapsack} 
even when generalized to arbitrary but fixed dimension $d$. In the \textsc{2DGeoKnapsack} where
we additionally interpret items as rectangles that need to be packed into a rectangular knapsack, we achieve a
$(\frac{9}{17}-\epsilon)$-competitive ratio. We also consider problem variants outside of the class of choosing problems and show that the static cases of \textsc{MaximumIndependentSet} for planar graphs with arriving edges and \textsc{MultipleKnapsack} also admit robust approximation schemes.
We  complement these positive results by also proving lower bounds for the
necessary migration  showing that some of our results are indeed tight. 
We also give lower bounds for different variants, including starting with an
adversarially chosen solution and 
different migration models.


\subsection{Related Work}
\paragraph*{Upper Bounds: } The general idea of bounded migration was introduced by
Sanders, Sivadasan, and Skutella~\cite{DBLP:journals/mor/SandersSS09}.
They developed an $(1+\epsilon)$-competitive algorithm with non-amortized
migration factor $f(1/\epsilon)$ for the \textsc{MakespanScheduling} problem.
G{\'{a}}lvez \emph{et al.}~\cite{DBLP:conf/esa/GalvezSV18} showed two
$(c+\epsilon)$-competitive algorithms with migration factor
$(1/\epsilon)^{O(1)}$ for some constants $c$ for the same problem. 
Skutella and Verschae~\cite{DBLP:journals/mor/SkutellaV16} were able to transfer
the $(1+\epsilon)$-competitive algorithm also to the setting, where items
depart. 
They also considered the \textsc{MachineCovering} problem and obtained an
$(1+\epsilon)$-competitive algorithm with amortized migration factor
$f(1/\epsilon)$. 
For \textsc{BinPacking}, Epstein and Levin~\cite{DBLP:journals/mp/EpsteinL09}
presented a $(1+\epsilon)$-competitive algorithm with non-amortized migration
factor $f(1/\epsilon)$.
Jansen and Klein~\cite{DBLP:conf/icalp/JansenK13} were able to obtain a
non-amortized migration factor of $(1/\epsilon)^{O(1)}$ for this problem, and
Berndt \emph{et al.}~\cite{DBLP:conf/approx/BerndtJK15} showed that such a
non-amortized migration factor is also possible for the scenario, where items
can depart.
Considering amortized migration, Feldkord \emph{et
  al.}~\cite{DBLP:conf/icalp/FeldkordFGGKRW18} presented a
$(1+\epsilon)$-competitive algorithm with migration factor $O(1/\epsilon)$ that
also works for departing items. 
Epstein and Levin~\cite{epstein2013cubepacking} investigated a multidimensional
extension of \textsc{BinPacking} problem, called \textsc{HypercubePacking} where
hypercubes are packed geometrically.
They obtained an $(1+\epsilon)$-competitive algorithm with worst-case migration
factor $f(1/\epsilon)$. 
For the preemptive variant of \textsc{MakespanScheduling}, Epstein and
Levin~\cite{DBLP:journals/algorithmica/EpsteinL14} obtained an exact online
algorithm with non-amortized migration factor $1+1/m$. 
Berndt \emph{et al.}~\cite{berndt2019bincovering} studied the
\textsc{BinCovering} problem with amortized migration factor and non-amortized
migration factor and obtained $(1+\epsilon)$-competitive algorithms and matching
lower bounds, even if items can depart.
Jansen \emph{et al.}~\cite{jansen2017strip} developed a
$(1+\epsilon)$-competitive algorithm with amortized migration factor
$f(1/\epsilon)$ for the \textsc{StripPacking} problem.
Finally, Berndt \emph{et al.}~\cite{DBLP:conf/waoa/BerndtDGJK19} developed a
framework similar to this work, but only for \emph{minimization problems}, and
showed that for a certain class of packing problems, any 
$c$-approximate algorithm can be combined with a suitable online algorithm to
obtain a $(c+\epsilon)$-competitive algorithm with amortized migration factor
$O(1/\epsilon)$. 

\paragraph*{Lower Bounds: }
Skutella and Verschae~\cite{DBLP:journals/mor/SkutellaV16} showed that a
non-amortized migration factor of $f(1/\epsilon)$ is not possible for any
function $f$ for \textsc{MachineCovering}. 
Berndt \emph{et al.}~showed that a non-amortized migration factor of
$\Omega(1/\epsilon)$ is needed~\cite{DBLP:conf/approx/BerndtJK15} for
\textsc{BinPacking}, and Feldkord \emph{et
  al.}~\cite{DBLP:conf/icalp/FeldkordFGGKRW18} showed that this also holds for
the amortized migration factor. 
Epstein and Levin~\cite{DBLP:journals/algorithmica/EpsteinL14} showed that exact
algorithms for the makespan minimization problem on uniform machines and for
identical machines in the restricted assignment setting have worst-case
migration factor at least~$\Omega(m)$.

\paragraph*{Dynamic Algorithms: }
In the semi-online setting, there are several metrics that one tries to
optimize.
First of all, there is the competitive ratio measuring the quality of the
solution.
Second, in order to prevent that the online problem simply degrades to the
offline problem, one needs to bound some resource.
In the setting that we consider, we bound the amount of repacking possible, as
such a repacking often comes with a high cost in practical applications.
In an alternate approach, often called \emph{dynamic algorithms}, we restrict
the running time needed to update a solution, ideally to a sub-linear function (see e.\,g.~the
surveys~\cite{DBLP:journals/rairo/BoriaP11,DBLP:conf/sofsem/Henzinger18}).  
Note that the amount of repacking used here can be arbitrarily high (using a
suitable representation of the current solution). 
This setting has been  also studied recently for
\textsc{Knapsack} variants~\cite{DBLP:journals/corr/abs-2007-08415} and 
\textsc{MaximumIndependentSet} variants~\cite{DBLP:conf/compgeom/Henzinger0W20}. 
There are also works aiming to combine both of the before mentioned
approaches~\cite{DBLP:conf/stoc/LackiOPSZ15,DBLP:conf/stoc/GuptaK0P17}.

\section{Upper Bounds for Choosing Problems}

\subsection{Framework for Choosing Problems}
\label{sec:choosing-problems}

In this section, we will consider the aforementioned class of \textit{choosing problems}. 
In general, a choosing problem is defined by a set of objects with some properties,
and the objective is to select a subset of these objects with maximum profit,
while potentially respecting some secondary constraints. 

\begin{definition}\label{Def:ChoosingP}
Consider a problem $\Pi$ where every instance $I\in \Pi$ is given by a set of
objects, where each object $i\in I$ is assigned a fixed profit value $p_i$, and
a set of feasible solutions $\SOL(I)$. We call $\Pi$ a choosing problem if
$\SOL(I)\subseteq \mathcal{P}(I)$, and the objective of some instance $I\in \Pi$ is to find a
subset $S\in \SOL(I)$ while maximizing the total profit $\PROFIT(S) = \sum_{i\in S}{p_i}$.
We further make the following two requirements for choosing problems that we will discuss in this paper:

\begin{enumerate}[(i)]
\item For any feasible solution $S \in \SOL(I)$, we have that any subset $S'
  \subseteq S$ is also a feasible solution, i.\,e.~$S'\in \SOL(I)$. 
\item For any solution $S\in \SOL(I_{t})$ for an instance $I_t$ with respective
  follow-up instance $I_{t+1}= I_t \triangle \{i_{t+1}\}$, the solution $S' := S
  \backslash \{i_{t+1}\}$ stays feasible for $I_{t+1}$, i.\,e.~$S'\in
  \SOL(I_{t+1})$. 
\end{enumerate}

For choosing problems, the profits of objects are their migration potential and costs; $\Delta(I_t)=p_i$, where~$i$ is the object added or removed at time $t$.
Given two solutions $S_1,S_2 \subseteq I$, we further have that $\phi(S_{1} \to S_{2}) = \PROFIT(S_1 \triangle S_2)$.

\end{definition}

While we restrict the range of problems with these properties,
it is necessary to do so since an adversary can enforce an unreasonably high migration factor for problems we excluded this way. In particular the first property guarantees that there is no low profit item with low migration potential added by the adversary which would allow for a completely new solution with high profit that we would need to switch to.
The second property serves a similar purpose, as it prevents the adversary from adding any arbitrary items making our current solution infeasible.

We could approach choosing problems like Berndt \emph{et al.}~\cite{DBLP:conf/waoa/BerndtDGJK19}
and use a greedy online algorithm in conjunction with the best known offline algorithm.
While this would also create good results,
we want to show that an online algorithm is not even necessary.
We propose instead the algorithm that computes an offline solution $S$ with profit $V$ and then waits until the total profit of items being added or removed from the instance exceeds $ \epsilon V$. At this point we simply replace $S$ with a new offline solution for the current instance. This algorithm already achieves a competitive rate close to the approximation ratio of the offline algorithm.

\begin{theorem}\label{ChoosingFW}
Let $A_{\textrm{off}}$ be an offline algorithm with an approximation ratio of $\alpha$.
Then the resulting framework using $A_{\textrm{off}}$ is a $(1-\epsilon)\cdot\alpha$-competitive algorithm requiring a migration factor of $O(1/(\epsilon\cdot\alpha))$.
\end{theorem}


\subsection{Resulting upper bounds and necessary migration}

The framework introduced can be applied to any choosing problem with some existing offline algorithm.
We observe that by using Theorem~\ref{ChoosingFW} for problems that admit a PTAS,
we get a respective robust PTAS,
and for other problems,
we achieve the ratio of any offline algorithm with small additional error.
We note without proof that all problems mentioned in the following theorem are choosing problems
as solutions are made up of sets of items or nodes
and they stay feasible under the required circumstances.
For a more detailed recap on these problems we refer to the appendix.

\begin{theorem}\label{MainResults}
For the following problems there exists an online algorithm with competitive ratio  $1-\epsilon$ and  migration factor $O(1/\epsilon)$.
\begin{itemize}
\item \textsc{SubsetSum} and \textsc{Knapsack} using the FPTAS by Jin~\cite{JinKnapsackPTAS}
\item  $d$-dimensional \textsc{Knapsack} using the PTAS from Caprara \emph{et al.} ~\cite{MultiDimKnapsackCapra}
\item \textsc{MaximumIndependentSet} on unweighted planar graphs by using the  PTAS by Baker~\cite{Baker1983}
\item \textsc{MaximumIndependentSet} on (weighted) $d$-dimensional disk-like
  objects with fixed $d$ using the PTAS by Erlebach \emph{et al.}~\cite{Erlebach2005}
\item \textsc{MaximumIndependentSet} on pseudo-disks using a PTAS by Chan and Har-Peled~\cite{Chan2009}
\end{itemize}

Additionally, for the \textsc{2DGeoKnapsack} problem, there exists an online algorithm with a competitive ratio $9/17-\epsilon$ and migration factor $O(1/\epsilon)$ by using the $9/17$ approximation algorithm from G{\'{a}}lvez \emph{et al.} \cite{2DGeoKnapsack}
\end{theorem}

We also show how to apply our framework in a setting outside of choosing problems.
We consider two problems still similar to the given context of \textsc{Knapsack} and \textsc{MaximumIndependentSet}.
\begin{theorem}
The static cases of \textsc{MultipleKnapsack} and \textsc{MaximumIndependentSet} on weighted planar graphs when adding edges both admit a robust PTAS with competitive ratio  $1-\epsilon$ and  migration factor $O(1/\epsilon)$.
\end{theorem}

Finally, one might ask how much migration is actually needed to achieve such results.
One can come up with quite simple lower bounds on the necessary migration for these kind of problems
by forcing the algorithm to switch between two different solutions.
If the adversary accomplishes this using little migration but with at least one solution being reasonably expensive to replace,
then this yields an interesting lower bound for \textsc{SubsetSum} and \textsc{MaximumIndependentSet}
showing that our framework is indeed tight in migration for all \textsc{Knapsack} variants and some \textsc{MaximumIndependentSet} variants.

\begin{theorem}
For the online \textsc{SubsetSum} problem and (weighted) \textsc{MaximumIndependentSet} problem,
there is an instance such that the migration needed for a solution with
value $(1-\epsilon)\OPT$ is $\Omega(1/\epsilon)$. 
\end{theorem}

\section{Conclusion}
In this paper, we present a general framework to transfer approximation algorithms for many
maximization problems to the semi-online setting with bounded migration.
Furthermore, we show that the algorithms constructed this way achieve optimal migration.
We expect our framework to be also applicable to other problems such as
\textsc{2DGeoKnapsack} variants with more complex objects~\cite{DBLP:conf/icalp/MerinoW20,DBLP:conf/esa/0001KW19},
\textsc{3DGeoKnapsack}~\cite{DBLP:journals/jcst/DiedrichHJTT08}, \textsc{DynamicMapLabeling}~\cite{DBLP:conf/esa/Bhore0N20},
\textsc{ThroughputScheduling}~\cite{DBLP:conf/ifipTCS/CieliebakEHWW04}, or \textsc{MaxEdgeDisjointPaths}~\cite{DBLP:journals/siamdm/ErlebachJ01}.

\bibliography{lib}
\newpage
\appendix

\section{Omitted Tables and Proofs} \label{omitted}
\subsection{Summary of Results}
\begin{center}
  \small
	\begin{threeparttable} 
\begin{tabular}{|c|c|c|c|}
\hline
\textsc{Problem} & Competitive Ratio & MF  & lower bound MF\\
\hline\hline
\textsc{Knapsack/SubsetSum} (fixed $d$) & $(1-\epsilon)$ & $O(1/\epsilon)$  & $\Omega(1/\epsilon)$\\
\hline
\textsc{2DGeoKnapsack} & $(9/17-\epsilon)$ & $O(1/\epsilon)$  & open \tnote{1} \\
\hline
\textsc{MIS} (unweighted planar) & $(1-\epsilon)$ & $O(1/\epsilon)$  & open\\
\hline
\textsc{MIS} (unweighted unit-disk) & $(1-\epsilon)$ & $O(1/\epsilon)$  & open\\
\hline
\textsc{MIS} (weighted disk-like objects\tnote{2}\ )  & $(1-\epsilon)$ & $O(1/\epsilon)$  & $\Omega(1/\epsilon)$\\
\hline
\textsc{MIS} on pseudo-disks & $(1-\epsilon)$ & $O(1/\epsilon)$  & open\\
\hline
\textsc{MultipleKnapsack}  & $(1-\epsilon)$ & $O(1/\epsilon)$  & $\Omega(1/\epsilon)$\\
\hline
\textsc{MIS} (unw. planar, edge arrival) & $(1-\epsilon)$ & $O(1/\epsilon)$  & open\\
\hline
\end{tabular}
\begin{tablenotes}\footnotesize 
\item[1] We prove a lower bound of $\Omega(1/\epsilon)$ for comp. ratio $(1-\epsilon)$.
\item[2] Result applies to all objects that the PTAS from Erlebach \emph{et al.}~\cite{Erlebach2005} can be used on, even in higher dimensions.
\end{tablenotes}
\end{threeparttable} 
\end{center}


\subsection{Proofs}

\begin{proof}[Proof of Theorem~\ref{ChoosingFW}]
In the following,
we refer to time steps where the offline algorithm is applied as repacking times.
Note that for the start of any online instance and the first arriving items,
we will just add the first eligible item with the offline algorithm when it arrives. This results in the first repacking time. Consider now any repacking time $t$ with generated solution $S_t$, and let $t'$ be the first point of time where $\Delta_{t:t'} > \epsilon \PROFIT(S_t)$. First we will show that the migration factor of the applied framework is small. We do so by proving that the migration potential between two repacking times accommodates for the repacking at the later repacking time. Denote with $A_{t:t'}$ the total profit of items that arrived up till time $t'$.

We now have that
\begingroup
\allowdisplaybreaks
\begin{align*}
&\frac{\phi(S_t \to S_{t'})}{\Delta_{t:t'}}\le\frac{\OPT(I_{t}) + \OPT(I_{t'})}{\Delta_{t:t'}}
\le\frac{\OPT(I_{t}) + \OPT(I_{t}) + A_{t:t'}}{\Delta_{t:t'}}
\le\\
&\frac{\OPT(I_{t}) + \OPT(I_{t}) + \Delta_{t:t'}}{\Delta_{t:t'}}
\le\frac{2 \cdot \OPT(I_{t})}{\Delta_{t:t'}} +1
\le\frac{2 \cdot \PROFIT(S_{t})}{\alpha \Delta_{t:t'}} +1 <\\
&\frac{2 \cdot \PROFIT(S_{t})}{\alpha \epsilon \PROFIT(S_t)} +1 \in O(1/(\epsilon\cdot\alpha)).
\end{align*}
\endgroup
Now it is left to show that up till time $t'-1$ we have maintained $(1-\epsilon)\alpha$ competitive rate. Denote with $A_{t:t'-1}$ and  $R_{t:t'-1}$ the total profit of items that arrived or departed up till time $t'-1$, and note that by choice of $t'$, we have that $A_{t:t'-1} +R_{t:t'-1} \le \epsilon \PROFIT(S_t)$. We now have 
\begin{align*}
&\PROFIT(S_{t'-1}) \ge \PROFIT(S_{t})- R_{t:t'-1}  
=\\
&\PROFIT(S_{t}) + A_{t:t'-1} - A_{t:t'-1}- R_{t:t'-1}\geq 
 \PROFIT(S_{t})+ A_{t:t'-1} - \epsilon \PROFIT(S_t) 
 =\\
 &(1-\epsilon) \PROFIT(S_{t})+ A_{t:t'-1}\geq 
 (1-\epsilon)\alpha \OPT(S_{t})+ A_{t:t'-1}\ge\\
 &(1-\epsilon)\alpha (\OPT(S_{t})+ A_{t:t'-1})\ge  (1-\epsilon)\alpha (\OPT(S_{t'-1})).
\end{align*}
\end{proof}

\begin{proof} [Proof of Theorem~\ref{MainResults}]
We prove that a $(1-\epsilon)$ approximation yields the desired result. The statement for \textsc{2DGeoKnapsack} follows similarly. Note that using a $(1-\epsilon)$ approximation with Theorem~\ref{ChoosingFW}, we achieve an online
algorithm with a competitive ratio of $(1-\epsilon)^2= 1-2\epsilon +\epsilon^2 $. By applying the framework with $\epsilon' = 1/2 \epsilon$,
we achieve the desired PTAS quality. The required migration of our
framework is bounded by $O(\frac{1}{\epsilon (1-\epsilon)}) = O(1/\epsilon)$. 
\end{proof}

\section{Upper Bounds}
\subsection{Knapsack Problems}

We will now discuss the problems in further detail. Note that \textsc{Knapsack} and the optimization variant of \textsc{SubsetSum} trivially are part of the class of choosing problems. The goal of these problems is to find a feasible subset, that is feasible in regards to the capacity constraint of the given Knapsack. Observe that both additional requirements for choosing problems are also fulfilled: Removing any item from a feasible solution will not make it infeasible as the new solution will have less total weight and in the same way adding some item $i$ will not make any prior infeasible subset of items feasible as it will only increase the total weight. Since \textsc{Knapsack} is a well studied problem we have access to a wide range of algorithms. Since we are looking to achieve a robust PTAS result, we will apply the FPTAS of Jin~\cite{JinKnapsackPTAS} and obtain the following result with our framework.

 \begin{theorem}

For the \textsc{Knapsack} and accordingly \textsc{SubsetSum} problem, there exists an online algorithm with competitive ratio  $1-\epsilon$ and  migration factor  $O(1/\epsilon)$.
\end{theorem}
\begin{proof}
By applying the given FPTAS with Theorem~\ref{ChoosingFW} we achieve an online
algorithm with a competitive rate of $(1-\epsilon)^2= 1-2\epsilon +\epsilon^2 $. We can therefore
apply the framework with $\epsilon' = 1/2 \epsilon$ yielding the desired solution quality. The required migration of our
framework is bounded by $O(\frac{1}{\epsilon (1-\epsilon)}) =O(1/\epsilon)$. 
\end{proof}

The $d$-dimensional \textsc{Knapsack} trivially falls in the class of choosing problems similar to its special case for $d=1$. For this problem, we use the PTAS from Caprara \emph{et al.} \cite{MultiDimKnapsackCapra}. As with the previous result we gain a similar result.

\begin{theorem}
For any fixed $d\in \mathbb{N}$, for the $d$-dimensional \textsc{Knapsack} problem there exists an online algorithm with  competitive ratio  $1-\epsilon$ and  migration factor  $O(1/\epsilon)$.
\end{theorem}

The geometric variants of, for example \textsc{2DGeoKnapsack}, add another layer to the problem by packing geometric objects into some actual knapsack.
In our framework however, this task is simply delegated to the respective offline algorithm. Therefore we can use the $\frac{9}{17} - \epsilon$ approximative algorithm from G{\'{a}}lvez \emph{et al.} \cite{2DGeoKnapsack} for the two-dimensional version of this problem that we denoted as \textsc{2DGeoKnapsack}.

\begin{theorem}
For the \textsc{2DGeoKnapsack} problem there exists an online algorithm with a competitive ratio $9/17-\epsilon$ and migration factor $O(1/\epsilon)$.
\end{theorem}





\subsection{Maximum independent set problems with incoming nodes}

As for \textsc{MaximumIndependentSet} being also regarded as a choosing problem,
we can firstly confirm that the properties from Section~\ref{sec:choosing-problems} hold.
That is, if we consider a problem instance $I$ and an independent set $S \subseteq I$,
then every subset $S' \subseteq S$ is still a set of non-adjacent nodes
and therefore remains an independent set.
Furthermore, a solution $S$ to the instance $I_t$ stays feasible for the next instance $I_{t+1}$
since no new edges connecting already present nodes are added.
In the following, we describe
how our framework can be applied to some variants of \textsc{MaximumIndependentSet}.

Generally, when speaking about the unweighted variants,
we assume the profit of every node (or item) in the instance to be 1.

\paragraph*{Maximum independent set for planar graphs}

While the maximum independent set is difficult to approximate~\cite{DBLP:journals/jacm/AroraLMSS98},
we limit ourselves to certain graph classes. 
This gives us the opportunity again to apply our framework on some efficient offline algorithms.
One graph class that is studied well in this context is the class of planar graphs,
that contains graphs that can be drawn on the plane without edge-intersections.
For \textsc{MaximumIndependentSet} in planar graphs,
we will use the PTAS by Baker~\cite{Baker1983} to achieve a robust PTAS similar to \textsc{Knapsack}.

\begin{theorem}
	For the \textsc{MaximumIndependentSet},  there exists an online algorithm with competitive ratio $1-\epsilon$ and  migration factor  $O(1/\epsilon)$.
\end{theorem}

\paragraph*{Maximum disjoint set for unit disks}


An instance for the \textsc{MaximumDisjointSet} problem can be easily transferred into a \textsc{MaximumIndependentSet} instance
by replacing every disk by a node and join two nodes by an edge if the corresponding disks are overlapping.
Also \textsc{MaximumDisjointSet} fulfils the aforementioned properties.
That is, if we consider a disjoint set $S \subseteq I$,
then every subset $S' \subseteq S$ is still a set of non-overlapping items and therefore a disjoint set.
Furthermore a solution $S$ to an instance $I_t$ stays feasible for the next instance $I_{t+1}$ since the position of items in $S$ is not changing and therefore remains a disjoint set.

For the offline algorithm solving \textsc{MaximumDisjointSet},
we follow the approach of Erlebach, Jansen and Seidel~\cite{Erlebach2005}
which gives us an PTAS for the offline (weighted) \textsc{MaximumIndependentSet} in disk graphs.

\begin{theorem}
	For the \textsc{MaximumIndependentSet} problem on (weighted) unit-disk graphs there exists a robust PTAS.
\end{theorem}

\paragraph*{MaximumDisjointSet for pseudo-disks}

For the \textsc{MaximumDisjointSet} problem limited to pseudo-disks,
we have another PTAS result in the offline world
introduced by Chan and Har-Peled~\cite{Chan2009}.
Using their PTAS,
we get yet another robust PTAS in the dynamic online world.

\begin{theorem}
	For the \textsc{MaximumDisjointSet} problem on pseudo-disks there exists a robust PTAS.
\end{theorem}


\section{Upper bounds for non-choosing problems}

There are many maximization problems that do not fit neatly under the umbrella
of choosing problems. However, it turns out that there is a wider class of problems that allow a similar treatment. In this section, we examine two problems that are not choosing problems (at least as defined in this paper): \textsc{MultipleKnapsack} and \textsc{MaximumIndependentSet} in weighted planar graphs with new edges added in each time step. The first problem is not a choosing problem essentially because it requires the selection of \emph{multiple} subsets of a given instance, not just one. The second problem is not a choosing problem because each time steps adds a new \emph{edge} to the graph instead of a \emph{vertex}.

\subsection{\textsc{MultipleKnapsack}}

It is known that the \textsc{MultipleKnapsack} problem has an offline PTAS~\cite{DBLP:journals/siamcomp/Jansen09,DBLP:conf/sofsem/Jansen12,DBLP:journals/siamcomp/ChekuriK05,DBLP:conf/random/Kellerer99}, which we will denote by $A_\text{off}$ in this subsection. On the basis of this PTAS, given $\epsilon$, we define the polynomial-time online algorithm $A^{\textsc{mk}}_\text{on}$ above.

\begin{lstlisting}[caption={Online algorithm $A^{\textsc{mk}}_{\textrm{on}}$ for \textsc{MultipleKnapsack}},backgroundcolor = \color{lightgray},label=list:8-6,captionpos=t,float,abovecaptionskip=-\medskipamount, escapeinside={(*}{*)},basicstyle=\ttfamily\linespread{1.15}\footnotesize]
compute a packing (*$S_1$*) for the first instance (*$I_1$*) using (*$A_\text{off}$*);
set (*$\Delta := 0$*) and (*$P := \PROFIT(S_1)$*);
for each time step (*$t-1 \to t$*) with (*$t > 1$*) and new item (*$i$*):
  set (*$S_t := S_{t-1}$*);
  increment (*$\Delta$*) by the profit (*$p_i$*) of (*$i$*);
  if (*$(\dagger) \ \Delta > \epsilon P$*):
    do a repacking with (*$A_\text{off}$*) to obtain the new solution (*$S_t$*);
    set (*$\Delta := 0$*) and (*$P := \PROFIT(S_t)$*);
  endif;
\end{lstlisting}

\begin{theorem}
The algorithm $A^{\textsc{mk}}_\textnormal{on}$ computes $(1-\epsilon)$-approximate solutions.
\end{theorem}

\begin{proof}
Whenever a repacking took place, the resulting solutions are $(1-\epsilon)$-approximate. If $t'$ is any time point at which no repacking took place, let $t$ be the latest time point before $t'$ at which a repacking did take place. In this case, we see that
\begin{align*}
&A^{\textsc{mk}}_\text{on}(I_{t'})
= \PROFIT(I_t) 
\geq (1-\epsilon) \OPT(I_t) 
\geq (1-\epsilon) (\OPT(I_{t'}) - \Delta_{t:t'}) \geq \\
& (1-\epsilon) (\OPT(I_{t'}) - \epsilon \cdot \PROFIT(I_{t'})) 
\geq (1-\epsilon) (\OPT(I_{t'}) - \epsilon \cdot \OPT(I_{t'})) 
\geq\\
&(1-\epsilon)^2 \OPT(I_{t'}) 
\geq (1-2\epsilon) \OPT(I_{t'}). 
\end{align*}
Thus, we see that, after a linear reparametrization of $\epsilon$, the algorithm $A^{\textsc{mk}}_\text{on}$ computes $(1-\epsilon)$-approximate solutions.
\end{proof}

\begin{theorem}
The migration factor of $A^{\textsc{mk}}_\textnormal{on}$ is bounded by $O(1/\epsilon)$.
\end{theorem}

\begin{proof}
It suffices to show the bound for individual phases in order to show the bound for all time intervals. Let $t \to \cdots \to t'$ be a phase. This implies that no repackings took place except at $t-1 \to t$ and $t' - 1 \to t'$. Then
\begin{align*}
&\gamma_{t:t'}
\leq \frac{\OPT(I_{t'})}{\Delta_{t:t'}} 
\leq \frac{\OPT(I_t) + \Delta_{t:t'}}{\Delta_{t:t'}}= 
 \frac{\OPT(I_t)}{\Delta_{t:t'}} + 1 \leq \\
& \frac{\OPT(I_t)}{\epsilon \PROFIT(I_{t'})} + 1 
\leq \frac{\OPT(I_t)}{\epsilon (1-\epsilon) \OPT(I_t)} + 1 = O(1/\epsilon),
\end{align*}
which proves the upper bound.
\end{proof}

\subsection{Maximum independent set in weighted planar graphs with edge arrivals}

We are given a weighted planar graph $G = (V,E)$ and a sequence of subgraphs
$(G_t)_{t=1}^T$ with $G_t = (V,E_t)$, for $E_t \subseteq E$. Moreover, we assume
that $|E_t \setminus E_{t-1}| = 1$ for all $t=2,...,T$, i.e. every time step
adds exactly one edge to the graph. We interpret this sequence as an online
instance of the static case of \textsc{MaximumIndependentSet}. Given $\epsilon \in
(0,\tfrac 12)$, we aim to find a polynomial-time online algorithm that computes
$(1-\epsilon)$-approximate solutions with a migration factor  $\mathcal{O}(1/\epsilon)$, where the migration potential of an edge is given by the minimum of the weights of the two adjacent vertices. Note that every graph $G_t$ in this sequence is planar, hence we know that there is an offline PTAS $A_\text{off}$ that can find $(1-\epsilon)$-approximate independent sets (IS) in polynomial time~\cite{Baker1983}. Using this, we define the online algorithm $A^{\textsc{IS}}_\text{on}$ above.

\begin{lstlisting}[caption={Online algorithm $A^{\textsc{IS}}_{\textrm{on}}$ for \textsc{MaximumIndependentSet} in planar graphs},backgroundcolor = \color{lightgray},label=list:8-6,captionpos=t,float,abovecaptionskip=-\medskipamount, escapeinside={(*}{*)},basicstyle=\ttfamily\linespread{1.15}\footnotesize]
compute an IS (*$S_1$*) for the first graph (*$G_1$*) using (*$A_\text{off}$*);
set (*$\Delta := 0$*) and (*$W := \WEIGHT(S_1)$*);
for each time step (*$t-1 \to t$*) with (*$t > 1$*) and new edge (*$e$*):
  if (*$e$*) has at most one adjacent vertex in (*$S_{t-1}$*):
    set (*$S_t := S_{t-1}$*);
  else:
    remove the smaller of the two adjacent vertices of (*$e$*)
    from (*$S_{t-1}$*) to obtain (*$S_t$*);
  endif;
  increment (*$\Delta$*) by the minimum of the weights of the two
  adjacent vertices of (*$e$*);
  if (*$(\dagger) \ \Delta > \epsilon W$*):
    use (*$A_\text{off}$*) to obtain the new solution (*$S_t$*);
    set (*$\Delta := 0$*) and (*$W := \WEIGHT(S_t)$*);
  endif;
\end{lstlisting}

\begin{theorem}
The algorithm $A^{\textsc{IS}}_\textnormal{on}$ computes $(1-\epsilon)$-approximate solutions.
\end{theorem}

\begin{proof}
Clearly, the approximation ratio is satisfied whenever a repacking took place. Thus, we let $t'$ be a time at which no repacking took place and $t$ the latest time before $t'$ where a repacking did take place. We see that
\begin{equation*}
\begin{aligned}
&A^{\textsc{IS}}_\text{on}(G_{t'})
= \WEIGHT(S_{t'}) 
\geq \WEIGHT(S_t) - \Delta_{t:t'} 
\geq_{(\dagger) } (1-\epsilon) \WEIGHT(S_t)\geq \\
& (1-\epsilon)^2 \OPT(G_t) \geq (1-\epsilon)^2 \OPT(G_{t'}) \geq (1-2\epsilon) \OPT(G_{t'}).
\end{aligned}
\end{equation*}
Thus, we have proven that, after a linear reparametrization of $\epsilon$, the algorithm $A^{\textsc{IS}}_\text{on}$ computes $(1-\epsilon)$-approximate solutions.
\end{proof}

\begin{theorem}
The migration factor of $A^{\textsc{IS}}_\textnormal{on}$ is bounded by $O(1/\epsilon)$.
\end{theorem}

\begin{proof}
To show the upper bound on the total migration factor, it suffices to show the upper bound for individual phases. To that end, let $t \to \cdots \to t'$ be a phase. In this case, we have
\begin{align*}
&\gamma_{t:t'}
\leq \frac{2 \cdot \OPT(I_t)}{\Delta_{t:t'}} 
\leq \frac{2 \cdot \WEIGHT(S_t)}{(1-\epsilon) \cdot \epsilon \cdot \WEIGHT(S_t)}
= O(1/\epsilon).
\end{align*}
This already implies the upper bound on the total migration factor.
\end{proof}

\section{Lower bounds and inapproximability results}

We will now complement our positive results by showing how much migration is necessary for some of these problems.

\subsection{General lower bounds for choosing problems}

We will begin for now with the dynamic version of these problems, where objects are added but may also be removed.
In this setting, the adversary is quite powerful because removing objects only leaves the option to repack our solution while the range of possibilities has become smaller.
This scenario is especially difficult when our current solution becomes inefficient and we might have to change the whole leftover solution.
An adversary can make sure that we might need to switch between two or more different solutions making a certain amount of migration necessary.
In order to prove our lower bounds of necessary migration, we want to create such a scenario.
In order to enforce switching between two solutions, we will use two independent instances, which we will define as \textit{alternating instances}:

\begin{definition}\label{def:alter}
Consider some dynamic online choosing problem $\Pi_{\textrm{on}}$, two instances $I_1$ and $I_2$, and some desired competitive ratio $\beta<1$.
We call the instances $I_1$ and $I_2$ alternating instances when there exists $I'_1 \subset I_1$,
such that the following properties hold:
\begin{itemize}
\item The solutions $S_1 =I_1$ and $S_2=I_2$ are feasible and henceforth also $S'_1 =I'_1$ is feasible, but any solution $S$ with $S \cap I_1 \neq \emptyset$ and $S \cap I_2 \neq \emptyset$ is infeasible.
\item We have that $\PROFIT(I'_1)< \beta \PROFIT(I_2)< \beta^2 \PROFIT(I_1)$.
\end{itemize}
 
\end{definition}

The idea behind these alternating instances is, that any algorithm that hopes to
be $\beta$-competitive can be forced to alternate between solutions of both
instances, when the adversary simply adds all objects from $I_1$ and $I_2$ and then continues to remove and add again all items from $I_1 \backslash I'_1$. We can further see that if the necessary migration for the solutions of $I'_1$ and $I_2$ is large but the migration potential from $I_1 \backslash I'_1$ is small, that this increases the ratio between necessary migration and migration potential.

\begin{lemma}
Let $\Pi_{\textrm{on}}$ be a dynamic online choosing problem and consider two alternating instances $I_1,I_2$ for a desired competitive rate $\beta<1$ with $I'_1\subseteq I_2$ fulfilling the requirements of definition \ref{def:alter}. When we additionally have that $\beta \PROFIT(I_2) \in \Omega(\gamma)$ for some desired migration factor $\gamma$ and $\PROFIT(I_1) - \PROFIT(I'_1) \le c$ for some $c$, then any $\beta$-competitive algorithm requires a migration factor of $\Omega(\frac{\gamma}{2(c+1)})$.
\end{lemma}
\begin{proof}

Consider some $\beta$-competitive algorithm $A$ and the following order of
events: First add all items from $I_1$ and  $A$ will generate some approximate
solution $S_1$ for $I_1$. Now, add also all items from $I_2$ and note that nothing changes. Since by definition $\PROFIT(I_2)< \beta \PROFIT(I_1)$, the algorithm does not need to do anything and will ignore the new items. Remember also that adding any of the new items from $I_2$ would make the solution infeasible. 

We now proceed to remove and add again all items of $I_1 \backslash I'_1$ and repeat this $N$ times for some large $N\in \mathbb{N}$. By removing all these items the previous solution $S_1$ would consequently be reduced to a solution $S'_1 \subseteq I'_1$, and as $\PROFIT(I'_1)< \beta \PROFIT(I_2)$, our algorithm needs to change to a solution $S_2 \subseteq I_2$. When the items are added again, the algorithm also needs to switch from $S_2$ to a solution of $I_1$.

Let us now look at the necessary migration and the migration potential. The total migration potential we received is given through the arrival of all items and the repeated removal and re-adding of items and altogether we have migration potential of $\PROFIT(I_1) + \PROFIT(I_2) + N(\PROFIT(I_1) - \PROFIT(I'_1)) = \PROFIT(I_1) + \PROFIT(I_2) + 2Nc$. The necessary migration results from the repacking of the solutions of $I'_1$ and $I_2$. We have to note however, that the necessary migration for the solutions of $I'_1$ might be small or even $0$, when the approximate solution $S_1$ for $I_1$ does not use any items of $I'_1$. We know however that for $I_2$ we at least exchange a full approximate solution which yields a total necessary migration of at least $N\beta \PROFIT(I_2)$. For the total migration factor, we now have that:

\begin{align*}
    &\frac{N \beta \PROFIT(I_2)}{\PROFIT(I_1) + \PROFIT(I_2) + 2Nc } \ge \frac{\beta \PROFIT(I_2)}{2(c+1)}
  \end{align*}
when $N$ is chosen large enough.
\end{proof}

We can conclude that one very natural way of proving lower bounds merely requires two instances, call them $I_1,I_2$, with certain properties. For one the adversary needs to have the possibility of being able to switch between instances $I_1$ and $I_2$ with low migration potential. If now for these two instances an algorithm create respective solutions $S_1,S_2$ in a way that switching between $S_1$ and $S_2$ becomes necessary, when the adversary switches between $I_1,I_2$, and changing solutions requires high migration then the mentioned algorithm will inevitably have a high migration factor.

In the following, we want to show how much repacking is necessary when we want to achieve a robust PTAS or rather a $(1-\epsilon)$-competitive algorithm for some iconic choosing problems.
We will start with the \textsc{SubsetSum} problem and show that our achieved migration of $O(1/\epsilon)$ is indeed optimal by proving a matching lower bound.

\subsection{Lower bounds on the migration factor for \textsc{SubsetSum}}

This section examines lower bounds and inapproximability results in various
cases of the \textsc{SubsetSum} problem to allow a fine-grained view on the
hardness of this online problem.
To this end, we split the analysis into four parts, depending on the static
vs. the dynamic case and whether the instances are \emph{lax}.
Being lax means that the first instance can be non-empty, i.\,e. that we can present a
non-empty instance fully at $t=1$ without contributing any
migration potential.
The original setting where we start with an empty instance is
the \emph{strict} case. 

We will show the following lower bounds on the migration factor of an online algorithm for the \textsc{SubsetSum} problem with ratio $1-\epsilon$ with $\epsilon \in (0,\tfrac 12)$:
\begin{center}
\begin{tabular}{|l|c|c|}
\hline
\textsc{SubsetSum} & strict & lax \\
\hline
static & $\Omega(\log 1/\epsilon)$ & $\Omega(1/\epsilon)$ \\
dynamic & $\Omega(1/\epsilon)$ & $\Omega(1/\epsilon)$ \\
\hline
\end{tabular}
\end{center}
Notice that the \textsc{SubsetSum} problem is a special case of the vanilla \textsc{Knapsack} problem, \textsc{MultipleKnapsack} and \textsc{2DGeoKnapsack}. This implies that the lower bounds also hold for these problems. In the preceding sections, we have seen that these problems can be solved using a total migration factor in $O(1/\epsilon)$, such that these bounds are tight except for the strict static case.

\begin{theorem}[Strict static case]
Any online algorithm for the strict static case of \textsc{SubsetSum} with ratio $1-\epsilon$ needs at least a migration factor in $\Omega(\log 1/\epsilon)$.
\end{theorem}

\begin{proof}
We construct an online instance of the strict static case of \textsc{SubsetSum} that generates a migration factor in $\Omega(\log 1 / \epsilon)$. To this end, let $T > 0$ be the number of time points in the instance and define the capacity $C := 2^T$. The first instance $I_1$ contains exactly one item $i_1$ of size $s_1 = 2^{T-1}$. For later time steps $t - 1 \to t$ with $t > 1$, we respectively add items $i_t$ of size $s_t = 3 \cdot 2^{T-t}$.

Consider the following sequence of \emph{alternating solutions}: Whenever $t$ is
odd, the solution consists of those items added in odd time steps, i.,e. we have
$S_t = \{ i_1,i_3,i_5,...,i_t \}$. When $t$ is even, the solution consists instead of those items added in even time steps, i.e. $S_t = \{ i_2,i_4,i_6,...,i_t \}$. They are visualized in the following diagram:
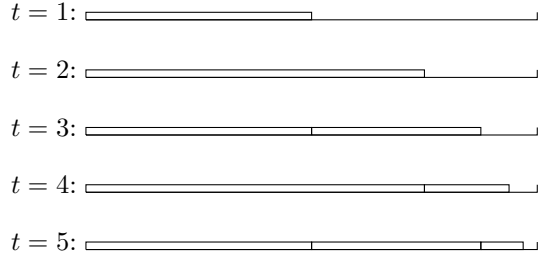
\begin{center}
$t=1$:
\begin{tikzpicture}
\draw (0,0) -- (6,0) -- (6,0.1);
\draw (0,0) -- (3,0) -- (3,0.1) -- (0,0.1) -- (0,0);
\end{tikzpicture}
\end{center}
\begin{center}
$t=2$:
\begin{tikzpicture}
\draw (0,0) -- (6,0) -- (6,0.1);
\draw (0,0) -- (4.5,0) -- (4.5,0.1) -- (0,0.1) -- (0,0);
\end{tikzpicture}
\end{center}
\begin{center}
$t=3$:
\begin{tikzpicture}
\draw (0,0) -- (6,0) -- (6,0.1);
\draw (0,0) -- (3,0) -- (3,0.1) -- (0,0.1) -- (0,0);
\draw (3,0) -- (5.25,0) -- (5.25,0.1) -- (3,0.1) -- (3,0);
\end{tikzpicture}
\end{center}
\begin{center}
$t=4$:
\begin{tikzpicture}
\draw (0,0) -- (6,0) -- (6,0.1);
\draw (0,0) -- (4.5,0) -- (4.5,0.1) -- (0,0.1) -- (0,0);
\draw (4.5,0) -- (5.625,0) -- (5.625,0.1) -- (4.5,0.1) -- (4.5,0);
\end{tikzpicture}
\end{center}
\begin{center}
$t=5$:
\begin{tikzpicture}
\draw (0,0) -- (6,0) -- (6,0.1);
\draw (0,0) -- (3,0) -- (3,0.1) -- (0,0.1) -- (0,0);
\draw (3,0) -- (5.25,0) -- (5.25,0.1) -- (3,0.1) -- (3,0);
\draw (5.25,0) -- (5.8125,0) -- (5.8125,0.1) -- (5.25,0.1) -- (5.25,0);
\end{tikzpicture}
\end{center}
and so on. In both cases, we see that
\begin{equation*}
\PROFIT(S_t) = 2^T - 2^{T-t}.
\end{equation*}
This is the largest number less than $C = 2^T$ that is divisible by $2^{T-t}$ and hence the best possible profit using items with sizes that are divisible by $2^{T-t}$. Note that the total capacity $2^T$ can never be reached, as all items except $i_1$ are divisible by three but $2^T$ is not. This implies that the sequence of alternating solutions described above is optimal and is in fact the only optimal sequence of solutions.

When examining approximate solutions, we observe that every algorithm with ratio $\alpha \in (0,1)$ is forced to choose this sequence of solutions as long as the inequality $\OPT(I_{t-1}) < \alpha \cdot \OPT(I_t)$ holds. Equivalently:
\begin{align*}
\OPT(I_{t-1}) < \alpha \cdot \OPT(I_t)
&\iff 2^T - 2^{T-t+1} < \alpha \cdot (2^T - 2^{T-t}) \\
&\iff 1 - 2^{-t+1} < \alpha - \alpha \cdot 2^{-t} \\
&\iff 2^{-t} (\alpha - 2) < \alpha - 1 \\
&\iff t < \log_2 \left( \frac{2-\alpha}{1-\alpha} \right).
\end{align*}
Hence, with $T = \lfloor \log_2((2-\alpha)/(1-\alpha)) - 1 \rfloor$, any algorithm with approximation ratio $\alpha$ will choose this solution. Its migration factor has the following lower bound:
\begin{align*}
\gamma
&\geq \frac{\sum_{t=1}^{T-1} \OPT(I_t)}{2^{T-1} + \sum_{t=2}^T 3 \cdot 2^{T-t}}
= \frac{\sum_{t=1}^{T-1} (2^T-2^{T-t})}{2^{T-1} + \sum_{t=2}^T 3 \cdot 2^{T-t}} \\
&= \frac{\sum_{t=1}^{T-1} (1-2^{-t})}{2^{-1} + 3 \cdot \sum_{t=2}^T 2^{-t}}
= \frac{T-1-\sum_{t=1}^{T-1} 2^{-t}}{2^{-1} + 3 \cdot \sum_{t=2}^T 2^{-t}} \\
&= \Omega(T)
= \Omega(\log \tfrac{2-\alpha}{1-\alpha}).
\end{align*}
In particular, with $\alpha = 1-\epsilon$, we get $\gamma = \Omega(\log 1/\epsilon)$.
\end{proof}

\begin{theorem}[Lax static case]
Any online algorithm for the lax static case of \textsc{SubsetSum} with ratio $1-\epsilon$ needs at least a migration factor in $\Omega(1/\epsilon)$.
\end{theorem}

\begin{proof}
Let the capacity $C$ be given by $\lfloor 1/\epsilon - 1 \rfloor$. Then
\begin{equation*}
\frac{C-1}{C} < 1-\epsilon.
\end{equation*}
Because we are in the lax case, we are free to present a full instance at $t=1$ without contributing migration potential. To exploit this, we let the first instance consist of an item $i$ of size $C-2$ and an item $j$ of size $C-1$. Because
\begin{equation*}
\frac{C-2}{C-1} < \frac{C-1}{C} < 1-\epsilon,
\end{equation*}
any algorithm with ratio $1-\epsilon$ will choose the solution consisting of $j$ only. For the next instance, let us insert an item $k$ of size $2$. After the insertion, the optimal solution will be $\{ i,k \}$ with profit $C$. The algorithm will follow this sequence of solutions, since
\begin{equation*}
\frac{C-1}{C} < 1-\epsilon.
\end{equation*}
This generates the migration factor
\begin{equation*}
\gamma = \frac{C-2 + C-1}{2} = \Omega(1/\epsilon).
\end{equation*}
\end{proof}

We now move on to the strict dynamic case. In the strict case, we are not allowed to present a full instance at $t=1$ without contributing migration potential. Instead, we have to start with an empty instance and add one item at a time. This destroys the argument from the lax static case. However, we can still salvage the basic parts to construct an instance for the strict dynamic case that generates a migration factor in $\Omega(1/\epsilon)$:

\begin{theorem}[Strict dynamic case]
Any online algorithm for the strict dynamic case of \textsc{SubsetSum} with ratio $1-\epsilon$ needs at least a migration factor in $\Omega(1/\epsilon)$.
\end{theorem}

\begin{proof}
Let the capacity $C$ be given by $\lfloor 1/\epsilon - 1 \rfloor$. Then, again,
\begin{equation*}
\frac{C-1}{C} < 1-\epsilon.
\end{equation*}
Add an item $i$ of size $C-2$. Clearly, the optimal solution contains exactly this item. Now, we add an item $j$ of size $C-1$. Because
\begin{equation*}
\frac{C-2}{C-1} < \frac{C-1}{C} < 1-\epsilon,
\end{equation*}
any algorithm with ratio $1-\epsilon$ will choose $\{ j \}$ as the solution. Let us now repeatedly insert and remove an item $k$ of size $2$. After each insertion, the optimal solution will be $\{ i,k \}$ with profit $C$. After each deletion, it will be $\{ j \}$ with profit $C-1$. The algorithm will follow this sequence of solutions, since
\begin{equation*}
\frac{C-1}{C} < 1-\epsilon.
\end{equation*}
Hence, after $N$ repetitions, we moved a load of $NC + N(C-1)$ while only a volume of $C-2 + C-1 + 4N = 2C-3+4N$ was inserted or removed. This generates the migration factor
\begin{equation*}
\gamma = \frac{NC + N(C-1)}{2C-3+4N}.
\end{equation*}
For $N \to \infty$, this gives $\gamma \geq (C-1)/4 = \Omega(1/\epsilon)$.
\end{proof}

\begin{corollary}
The same theorem holds for the lax dynamic case of \textsc{SubsetSum}.
\end{corollary}


\subsection{Inapproximability of \textsc{Knapsack} with weight migration}

In the preceding sections, we have assumed that migration is measured by total profit and the migration potential is similarly given by the profit of new items. In this section, we give inapproximability results for the case in which migration is measured by total \emph{weight} and the migration potential is given by the \emph{weight} of new items. Intuitively, we exploit that the \textsc{Knapsack} problem is only concerned with maximizing profit, which is a priori uncorrelated with the weight of a solution. It is thus possible to trick an algorithm into migrating a lot of weight for an item with large profit but small weight, i.e. small migration potential.

\begin{theorem}[Lax static case]\label{knapsack_static_case_by_weight}
There cannot exist an online algorithm with bounded migration factor for the lax static case of \textsc{Knapsack} with approximation ratio $\alpha > 1/2$ when migration is measured by weight.
\end{theorem}

\begin{proof}
We construct an online instance of the lax static case that is unsolvable under the hypothesis of a bounded migration factor.

Suppose that there is an online algorithm $A_\text{on}$ that solves \textsc{Knapsack} with an approximation ratio of $\alpha > 1/2$ and assume that the migration factor of $A_\text{on}$ is bounded by some constant $B$. Let $C > B$ be the capacity of the knapsack. The instance $I_1$ at $t=1$ contains the item $i$ with weight $C$ and profit $1$. Our algorithm $A_\text{on}$ will choose the solution $S_1 = \{ i \}$. The instance $I_2$ at $t=2$ is constructed by adding the item $j$ to $I_1$ with weight $1$ and profit $2$. Due to $\alpha > 1/2$, this forces the algorithm to choose the solution $S_2 = \{ j \}$ and the migration factor
\begin{equation*}
\gamma = C > B,
\end{equation*}
is generated. This is in contradiction with the assumption that $B$ be an upper bound on the migration factor.
\end{proof}

\begin{corollary}
The same theorem holds for the following generalizations of the lax static case of \textsc{Knapsack}:
\begin{itemize}
\item The lax dynamic case of \textsc{Knapsack}.
\item The lax static and dynamic cases of \textsc{MultipleKnapsack}.
\item The lax static and dynamic cases of \textsc{2DGeoKnapsack}.
\end{itemize}
\end{corollary}

Under the hypothesis that the first instance be empty, the construction in the proof of the previous theorem fails. However, we can reuse the basic ideas for a counterexample in the strict dynamic case and obtain:

\begin{theorem}[Strict dynamic case]
There cannot exist an online algorithm with bounded migration factor for the strict dynamic case of \textsc{Knapsack} with approximation ratio $\alpha > 1/2$ when migration is measured by weight.
\end{theorem}

\begin{proof}
Let $i$ and $j$ be the same items as in the proof of theorem \ref{knapsack_static_case_by_weight}.
We construct an online instance as follows: The first instance is empty as assumed. The second instance $I_2$ is given by $\{ i \}$. In the following time steps, we successively add and remove $j$ to and from the instance for a total of $N$ times. The algorithm has to follow the sequence of optimal solutions by the argument in the proof of theorem \ref{knapsack_static_case_by_weight}. This generates the migration factor
\begin{equation*}
\gamma = \frac{NC}{C+N},
\end{equation*}
which is unbounded.
\end{proof}

\begin{corollary}
The same theorem holds for the following generalizations of the strict dynamic case of \textsc{Knapsack}:
\begin{itemize}
\item The lax dynamic case of \textsc{Knapsack}.
\item The lax and strict dynamic cases of \textsc{MultipleKnapsack}.
\item The lax and strict dynamic cases of \textsc{2DGeoKnapsack}.
\end{itemize}
\end{corollary}

\subsection{Lower bounds for maximum independet set}

We now take a look at \textsc{MaximumIndependentSet} again. If we consider \textsc{MaximumIndependentSet} on arbitrary graphs we can actually prove the same bound for necessary migration, when aiming for a robust PTAS. Since we can choose any selection of edges among nodes in the graph we can simply emulate the same instances that we constructed for the \textsc{SubsetSum} Problem.

\begin{theorem}
  There is an instance of the online
  \textsc{MaximumIndependentSet} problem such that the migration needed for a solution with
  value $(1-\epsilon)\OPT$ is $\Omega(1/\epsilon)$. 
\end{theorem}
\begin{proof}
Set $C := 2\lceil 1/(3\epsilon) \rceil$ and consider two sets of nodes $V_1,V_2$ with edges $E :=\{{v_1,v_2}| v_1\in V_1,v_2\in V_2\}$, then both $V_1$ and $V_2$ are feasible solutions while mixing them would destroy independence. Set $V'_1 := V_1 \backslash \{v,w\}$ for two nodes $v,w\in V_1$ and by adding all nodes and then repeatedly removing $v$ and $w$ and re-adding them we created the same situation as for \textsc{SubsetSum} leading to a necessary migration factor of $\Omega(1/\epsilon)$.  
\end{proof}

While this construction works on arbitrary graphs the same construction can be difficult or rather impossible if we limit ourselves to certain graph classes. If for example we consider \textsc{MaximumIndependentSet} on planar graphs the same instance cannot be built. Since we have more than three nodes in each of the two sets $V_1,V_2$ the graph would have to contain a $K_{3,3}$ as a subgraph and would therefore not be a planar graph. A similar argument is true for unit-disk graphs. When given one node $v$ which we can consider without loss of generality to be in set $V_1$, then we can see that the number of nodes we can add connected to $v$ but independent to each other is bound by the number of unit-disks we can place non intersecting such that all their center points lie on a disk with radius $2$ around the center of $v$. This number is bound and since we can choose $\epsilon$ small enough such that $|V_2|$ exceeds that bound, constructing an instance like above is not possible. This leaves an open problem whether the migration that our algorithm achieved for the cases of \textsc{MaximumIndependentSet} on planar or unit disk graphs can be improved or whether there exists instances that enforce a migration of $O(1/\epsilon)$. 

In the world of weighted \textsc{MaximumIndependentSet} however it is again possible to prove the lower bound even for graph classes such as unit-disks or planar graphs. We can in this setting once again emulate the same instance that we had for the unweighed case, but can drastically reduce the number of necessary nodes by introducing nodes with high weight, leading to the following result.

\begin{theorem}
  There is an instance of the online weighted
  \textsc{MaximumIndependentSet} problem, containing only three nodes on one path, such that the migration needed for a solution with
  value $(1-\epsilon)\OPT$ is $\Omega(1/\epsilon)$. 
\end{theorem}
\begin{proof}
Set $C := 2\lceil 1/(3\epsilon) \rceil$ and consider the set of nodes $V=\{v_1,v_2,v_3\}$ with weights $w_1 := C-2, w_2 := C-1,w_3:= 2$ and $E:= \{\{v_1,v_2\}\{v_2,v_3\}\}$. Then the online instance where we add all three nodes $v_1,v_2,v_3$ and then repeatedly remove and add $v_3$ is the required instance.
\end{proof}

We can further conclude that this construction works on any graph class that admits a path of length three and therefore solving this problem with a competitive rate of $1-\epsilon$ for $\epsilon <1$ requires a migration factor of $\Omega(1/\epsilon)$.

\subsection{Inapproximability of maximum independet set in unit disk graphs with area migration}

We consider the problem of finding maximum independet sets in unweighted unit disk graphs and show that the corresponding online problem cannot be solved with an approximation ratio $\alpha > 1/2$ under the hypothesis of a bounded migration factor \emph{if we assume that the migration potential in the time step $t-1 \to t$ is given by the difference in area} between consecutive graphs:
\begin{equation*}
\Delta(I_t) = \AREA(I_t) - \AREA(I_{t-1}).
\end{equation*}

\begin{theorem}[Lax static case]
The lax static case of \textsc{MaximumIndependentSet} in unit disk graphs cannot be solved by an algorithm with ratio $\alpha > 1/2$ when the migration potential is given by the area difference of consecutive graphs.
\end{theorem}

\begin{proof}
We construct an online instance that needs an unbounded migration factor in order to be solved. The first graph consists of two unit disks $D_1$ and $D_2$ with $D_1$ tangent to $D_2$ on the left. Note that the disks overlap. Thus, they can never be part of the same solution. Due to symmetry, we can assume without loss of generality that any algorithm with a ratio $\alpha > 1/2$ will choose the solution $\{ D_1 \}$. Now, an adversary could add the unit disk $D_3$ slightly to the left of $D_1$ in such a way that the area difference (the migration potential) is exactly equal to $\delta > 0$. The situation is summarized in the following picture:
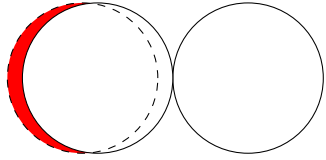
\begin{center}
\begin{tikzpicture}
\draw (2,2) circle (1cm);
\fill[red] (-0.2,2) circle (1cm);
\fill[white] (0,2) circle (1cm);
\draw (0,2) circle (1cm);
\draw[dashed] (-0.2,2) circle (1cm);
\end{tikzpicture}
\end{center}
Any approximation algorithm with a ratio $\alpha > 1/2$ will have to choose the solution $\{ D_1,D_3 \}$, which generates the migration factor $\gamma = 2 \pi / \delta$. This is unbounded for $\delta \to 0$.
\end{proof}

\begin{theorem}[Strict dynamic case]
The strict dynamic case of \textsc{MaximumIndependentSet} in unit disk graphs cannot be solved by an algorithm with ratio $\alpha > 1/2$ when the migration potential is given by the area difference of consecutive graphs.
\end{theorem}

\begin{proof}
In analogy to the lax static case, we construct an instance of the strict dynamic case that generates an unbounded migration factor. The instance is constructed as follows: Add a unit disk $D_1$. Any approximation algorithm with ratio $\alpha > 1/2$ will choose $\{ D_1 \}$ as its solution. Then, we add $D_2$ tangent to $D_1$ on the right:
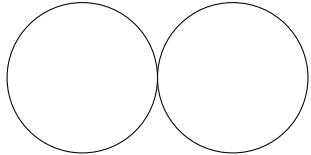
\begin{center}
\begin{tikzpicture}
\draw (0,2) circle (1cm);
\draw (2,2) circle (1cm);
\end{tikzpicture}
\end{center}
To minimize migration, an online algorithm will choose $\{ D_1 \}$. Again, the adversary could add a unit disk $D_3$ slightly to the left of $D_1$ in such a way that the area difference is exactly equal to $\delta > 0$. Now, any algorithm with a ratio $\alpha > 1/2$ has to choose $\{ D_2,D_3 \}$ as the solution, such that both $D_1$ and $D_2$ have to migrate. The adversary then removes $D_3$ from the instance, generating a migration potential of $\delta$ and the algorithm is again forced to choose either $D_1$ or $D_2$. To minimize migration, the algorithm would have to choose $D_2$. In this case, the adversary adds $D_4$ slightly to the right of $D_2$ and we would have the same analysis as before. When we play this game for a total of $N$ times, this gives the migration factor
\begin{equation*}
\gamma = \frac{2 \pi N}{2 \pi + 2 \delta N}.
\end{equation*}
With $\delta = 1/N$ and $N \to \infty$, this blows up!
\end{proof}

\section{Framework with complementing Online Algorithm}\label{GenFW}

As we have seen so far, we can solve a range of classical and famous problems in their online versions with small migration by only using an online algorithm. While we have no meaningful result or applications for it, we now want to discuss how to use two algorithms in an alternating fashion and how well the resulting combination performs in terms of migration and competitive rate. The general idea resembles the framework of Berndt \emph{et.  al.}. \cite{DBLP:conf/waoa/BerndtDGJK19} and generalizes the framework we have used so far. Instead of simply waiting between the applications of the offline algorithm in our framework, we now want to bridge the time by applying some known online algorithm. In this way we hope to achieve and use the best results of both worlds: the flexibility of online algorithms and the high solution quality from the offline world.To make this idea work however, we cannot consider any combination of algorithms and we need to ensure that our online algorithm is able to work with the solutions of the offline algorithm. We will call such algorithms \textit{flexible}.




\begin{definition}
Let $I \in \Pi_{\textrm{on}}$ be an instance of the online problem and an online algorithm $A_{\textrm{on}}$ for this problem. Let $t<t' \le |I|$ be two points of time and $S_t$ be a solution for $I_t$ not necessarily generated by $A_{\textrm{on}}$. We say an online algorithm is flexible, if it also accepts $S_t$ as a parameter and extends the solution $S_t$ to a solution $S_{t'}$ for $I_{t'}$ by reacting to the events happening in the time interval $t \to \cdots \to t'$. We further say $A_{\textrm{on}}$ has a maintaining ratio of $\beta$ from $t$ to $t'$, when $A_{\textrm{off}}(I_t) = \PROFIT(S_t) \ge \alpha \OPT(I_t)$ implies that $A_{\textrm{on}}(I_{t'},S_t) = \PROFIT(S_{t'}) \ge \alpha \beta \OPT(I_{t'})$.
\end{definition}

When we combine two algorithms, both delivering approximate solutions, it is inevitable that the solution quality will deteriorate based on both algorithms. We want to try and achieve a final ratio of $\alpha \cdot \beta$ where $\alpha$ is the best known offline approximation ratio and $\beta = 1-O(\epsilon)$ in order to achieve a similar competitive ratio to the offline result, except a small error of $O(\epsilon)$. We acknowledge therefore that our online algorithm may not uphold this ratio permanently but over some time frame. In this time frame, up to the earliest point of time where the online algorithm would break this desired ratio, we have to exchange our solution for a new better one. On the other side, we also want to achieve a certain migration factor. Therefore, we need to wait long enough until there is a time where the migration costs of exchanging our solution and the migration potential of newly arrived objects or information balances each other out. If such a point of time exists in the time frame where we maintain our desired solution quality, we call this time point a \textit{repacking time}. 

\begin{definition}
Let $\beta,\gamma\in \mathbb{Q}_{> 0}$ and let $A_{\textrm{on}}$ be a flexible online algorithm for $\Pi_{\textrm{on}}$. Let $t$ be any point of time, $S_t$ some solution for $I_t$. Let $t'$ be the first point of time, where $A_{\textrm{on}}$ is not able to keep the maintaining ratio $\beta$. If we then have for some $t< t^{\star} \le t'$ that for some other $\beta$-competitive solution $S'_{t^{\star}}$ that $\frac{\phi(S_{t^{\star}-1} \to S'_{t^{\star}})}{\Delta_{t:t^{\star}}} \le \gamma$, we call $t^{\star}$ a $(\beta,\gamma)$-repacking time and say that $A_{\textrm{on}}$ starting with solution $S_t$ admits a repacking time with maintaining ratio of $\beta$ and  migration  $\gamma$ for~$S'_{t^{\star}}$.
\end{definition}

This definition is quite powerful, and we do not require these properties for any arbitrary solutions. For our framework it is important that our chosen online and offline algorithms can work cooperatively. We require that the online algorithm starting with an offline solution maintains the desired competitive ratio until the offline algorithm computes another solution that we can afford migrating to. In that sense, we introduce the term of \textit{compatibility}. If an online algorithm working on a solution of an offline algorithm always admits a $(\beta,\gamma)$-repacking time for some future solution of the offline algorithm, we then call both algorithms \textit{compatible}.

\begin{definition}
Let $A_{\textrm{on}}$ and $A_{\textrm{off}}$ be an online and an offline algorithm for $\Pi$. We say $A_{\textrm{on}}$ and $A_{\textrm{off}}$ are compatible with maintaining ratio $\beta$ and migration factor $\gamma$ if $A_{\textrm{on}}$ starting with some solution $S$ from $A_{\textrm{off}}$ admits a $(\beta,\gamma)$-repacking time for some future solution $S'$ of $A_{\textrm{off}}$.
\end{definition}

We will often mention the time frame from one repacking time until the next and we will regard such a time window as \emph{phase}. During a phase, we will use the online algorithm and handle the changes of the instance until a repacking time occurs. When this happens, we basically want to switch to the solution of the offline algorithm. It may happen that we need to exchange parts or maybe even the complete old solution. The combination of two compatible algorithms achieves a competitive rate dependent on both the approximation ratio of the offline algorithm and the maintaining ratio.

\begin{theorem}
Let $A_{\textrm{off}}$ be an offline algorithm with an approximation ratio of $\alpha$ and $A_{\textrm{on}}$ be an online algorithm compatible with $A_{\textrm{off}}$ maintaining ratio $\beta$ and migration factor $\gamma$. Then the combined algorithm $F$ is an online algorithm with a competitive rate of $\alpha \cdot \beta$ and migration factor $\gamma$.

\end{theorem}

\begin{proof}
Note that at a $(\beta,\gamma)$-repacking time $t^{\star}$, we obviously have, due to the approximation ratio of $A_{\textrm{off}}$, that $F(I_{t^{\star}}) = A_{\textrm{off}}(I_{t^{\star}})\ge \alpha \OPT(I_{t^{\star}})$. Similarly, we have for any non-repacking time $t$, whose last previous repacking time is $t^{\star}$, that $F(I_{t}) = A_{\textrm{on}}(I_{t}, S_{t^{\star}})\ge \alpha \beta \OPT(I_{t})$ due to the definition of $A_{\textrm{on}}$. For the migration factor, we consider each phase from one repacking time $t^{\star}_1$ to $t^{\star}_2$ and w.l.o.g. consider $t=1$ to be the first repacking time as we can simply start with an offline solution. By definition of compatibility with maintaining ratio and migration factor, we know that $\frac{\phi(S_{t^{\star}_2-1} \to S_{t^{\star}_2})}{\Delta_{t^{\star}_1:t^{\star}_2}} \le \gamma$. For the final migration factor, consider any point of time $t$ and say that $t$ lies after $k$ repacking times. Denote with $t^{\star}_i$ those repacking times and with  $S_{t^{\star}_{i}-1}$ the respective repacked solutions and with $S_{t^{\star}_{i}}$ the solutions after each repacking for $1\le i \le k$, where $S_{0}$ is the empty solution. We then have that our migration factor is bounded by
$ \frac{\sum_{i=2}^{k}{\phi(S_{t^{\star}_{i}-1} \to S_{t^{\star}_{i}})}}{\Delta_{0:t}}  
\le \frac{\sum_{i=2}^{k}{\phi(S_{t^{\star}_{i}-1} \to S_{t^{\star}_{i}})}}{\sum_{i=1}^{k}{\Delta_{t^{\star}_{i-1}:t^{\star}_{i}}}} \le \gamma$.
Note that we start summing up migration costs at $i =2$ since at the first repacking time we simply start with an offline solution and hence have no migration costs.
\end{proof}

Overall we end up with a very simple framework. All we need is two algorithms that we apply in an alternating fashion as long as we are able to repack our solutions and balance migration potential and migration costs. 




\end{document}